\documentclass[aps,pra,twocolumn,superscriptaddress]{revtex4-2}
\usepackage{amsmath,amssymb,bm,physics,graphicx,hyperref,cleveref,siunitx,booktabs,amsthm}
\usepackage[mathscr]{euscript}
\usepackage{appendix}
\usepackage{tikz}
\usetikzlibrary{positioning,arrows,shapes}

\DeclareMathOperator{\erfc}{erfc}
\newcommand{\SNR}{\mathrm{SNR}}
\newcommand{\calC}{\mathcal{C}}
\newcommand{\calN}{\mathcal{N}}
\newcommand{\calB}{\mathcal{B}}

\newcommand{\SigmaDotAvg}{\langle\dot{\Sigma}\rangle}
\newcommand{\betaCold}{\beta_{\mathrm{cold}}}
\newcommand{\betaHot}{\beta_{\mathrm{hot}}}
\newcommand{\betaBar}{\bar{\beta}}
\newcommand{\DeltaBeta}{\Delta\beta}
\newcommand{\epsNet}{\epsilon_{\mathrm{net}}}
\newcommand{\epsMin}{\epsilon_{\min}}

\newcommand{\SigmaTot}{\Sigma_{\mathrm{tot}}}
\newcommand{\betaIn}{\boldsymbol{\beta}_{\mathrm{in}}}
\newcommand{\betaZ}{\beta_z}
\newcommand{\betaZInf}{\beta_z^\infty}
\newcommand{\betaV}{\beta_v}

\newcommand{\kB}{k_B}

\theoremstyle{plain}
\newtheorem{theorem}{Theorem}

\theoremstyle{definition}

\begin{document}

\title{The Thermodynamic Cost of Computing with Heat}

\author{M. W. AlMasri}
\affiliation{Wilczek Quantum Center, School of Physics and Astronomy, Shanghai Jiao Tong University, Minhang, Shanghai, China}
\email{mwalmasri2003@gmail.com}

\date{\today}

\begin{abstract}
Autonomous quantum thermal machines have recently been proposed as physics-based computing substrates where logical inputs and outputs are encoded in temperature gradients. While such ``thermodynamic neurons'' exhibit a clear trade-off between computational fidelity and heat dissipation, the fundamental information-theoretic limits of temperature-encoded computation remain uncharacterized. Here, we derive rigorous bounds linking average error probability, channel capacity, and entropy production for finite-capacity thermal reservoirs operating far from equilibrium. We prove that the minimal dissipation required to achieve a target average error probability $\langle \xi \rangle$ diverges as $\langle \xi \rangle$ approaches a fundamental minimum error floor $\epsMin$ imposed by finite-reservoir thermal fluctuations. We further establish a thermodynamic channel capacity that saturates at high dissipation, and quantify the minimal dissipation required for cascaded networks to maintain target fidelity. We demonstrate that the required dissipation grows linearly $\mathcal{O}(L)$ with network depth for shallow networks, but diverges as the network depth $L$ approaches a fundamental maximum limit $L_{\max} = \epsNet/\epsMin$ imposed by the minimum error floor. Furthermore, under strong noise amplification conditions, the required dissipation can grow up to $\mathcal{O}(L^3)$. Our framework bridges stochastic thermodynamics, finite-time information theory, and autonomous computation, providing rigorous design principles for energy-efficient analog thermodynamic hardware.
\end{abstract}

\maketitle

\section{Introduction}
The convergence of non-equilibrium thermodynamics and information processing has catalyzed the emergence of \emph{thermodynamic computing}~\cite{thermocomp1,Esposito,Ueda,thermocomp2,thermocomp3,Wolpert,Chattopadhyay2025}, a paradigm that exploits dissipative physical dynamics to perform logical and machine-learning tasks. Recent comprehensive reviews have highlighted the fundamental role of the Landauer principle and its extensions in understanding the thermodynamics of computation, particularly in the quantum regime with finite-size heat baths and non-Markovian environments~\cite{Chattopadhyay2025}. Recent work introduced \emph{thermodynamic neurons}---autonomous quantum thermal machines comprising few interacting qubits coupled to thermal baths at distinct temperatures~\cite{lipkabartosik2024}. Logical inputs are encoded in bath temperatures, while outputs are read from the steady-state temperature of a finite-capacity reservoir. These devices implement linearly separable Boolean functions and, when cascaded, achieve universal classical computation. Crucially, they operate without external control, relying solely on heat flows and virtual temperature engineering.

Despite their conceptual elegance, the information-theoretic performance limits of temperature-encoded computation remain unexplored. Conventional bounds such as Landauer's principle~\cite{landauer} address the thermodynamic cost of logical erasure in near-equilibrium, Markovian settings. In contrast, thermodynamic neurons function far from equilibrium, encode information in continuous macroscopic variables, and rely on finite reservoirs that intrinsically introduce stochastic fluctuations. The interplay between dissipation, noise robustness, and information capacity in this regime is governed by principles that extend beyond classical information theory and require a stochastic thermodynamic formulation.

In this paper, we establish rigorous information-thermodynamic bounds for temperature-encoded computation. We (i) formalize the thermodynamic neuron as a noisy thermodynamic channel, (ii) derive a strict trade-off between average error probability and entropy production, revealing a fundamental minimum error floor, (iii) prove an upper bound on the thermodynamic channel capacity that saturates at high dissipation, and (iv) quantify the minimal dissipation overhead required for cascaded networks to maintain target fidelity, revealing a linear $\mathcal{O}(L)$ scaling for practical depths and a fundamental divergence as the network depth approaches a maximum limit set by the minimum error floor. Our results provide rigorous constraints for the design of autonomous analog computing hardware and clarify the energetic price of temperature-based information processing. These bounds complement recent work on resilience-runtime tradeoffs in quantum algorithms~\cite{GarciaPintos2025}, extending the understanding of fundamental limits to autonomous thermal computation.

\section{Model and Information-Theoretic Framework}
We consider a thermodynamic neuron with $n$ input baths at inverse temperatures $\betaIn = (\beta_1,\dots,\beta_n)$, a reference bath at $\beta_0$, and an output reservoir $\calB_z$ of finite heat capacity $C$. The machine evolves under a time-independent Hamiltonian and weak local dissipation, reaching a non-equilibrium steady state (NESS) characterized by a virtual inverse temperature
\begin{equation}
\betaV = \sum_{k=0}^n w_k \beta_k, \quad w_k = \frac{h_k \epsilon_k}{\sum_{j=0}^n h_j \epsilon_j},
\label{eq:virtual}
\end{equation}
where $\epsilon_k$ are qubit energy gaps and the indicators $h_k\in\{0,1\}$ specify the interaction subspace (with $h_j\epsilon_j \geq 0$ by construction). The output reservoir temperature $T_z(t)$ (with inverse temperature $\betaZ(t) = 1/(k_B T_z(t))$) obeys the calorimetric equation $C\dot{T}_z = \sum_k J_k$, where $J_k$ denotes heat currents flowing into the reservoir. In the steady-state limit $t\to\infty$, the output converges to $\betaZInf = f(\betaV)$, with $f$ a sigmoid-like activation determined by modulator parameters~\cite{lipkabartosik2024}. 

It is crucial to distinguish between the \emph{bath temperatures} ($\beta_k$), the \emph{virtual temperature} ($\betaV$) which characterizes the steady-state population of the machine's internal states, and the \emph{effective reservoir temperature} ($\betaZInf$), which is the macroscopic observable from which the output is decoded.

Logical inputs $x\in\{0,1\}^n$ map to inverse temperatures via $\beta_k = \betaHot + x_k(\betaCold-\betaHot)$, where $x_k \in \{0,1\}$ is the $k$-th bit of the input string. 
Outputs are decoded by thresholding $\betaZInf$ against the midpoint $\betaBar = (\betaCold+\betaHot)/2$: $y=0$ if $\betaZInf < \betaBar$, and $y=1$ otherwise. Finite $C$ induces Gaussian fluctuations $\betaZInf \sim \calN(f(\betaV),\sigma^2)$, with variance $\sigma^2$ arising from both thermal and dynamical sources.

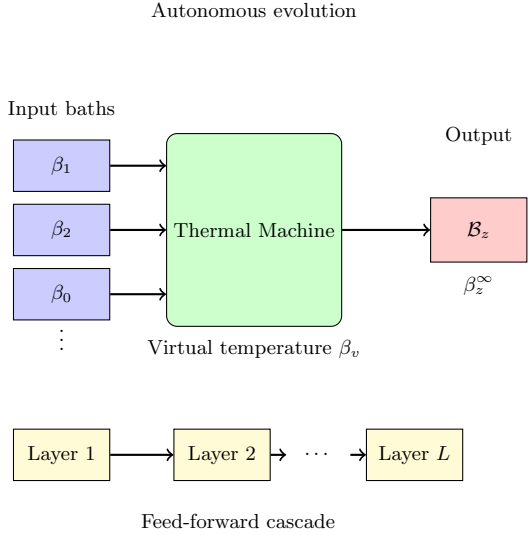
\begin{figure}[t]
\centering
\begin{tikzpicture}[scale=0.85, every node/.style={transform shape}]
\node[rectangle, draw, fill=blue!20, minimum width=1.5cm, minimum height=0.8cm] (bath1) at (0,2) {$\beta_1$};
\node[rectangle, draw, fill=blue!20, minimum width=1.5cm, minimum height=0.8cm] (bath2) at (0,1) {$\beta_2$};
\node[rectangle, draw, fill=blue!20, minimum width=1.5cm, minimum height=0.8cm] (bath0) at (0,0) {$\beta_0$};
\node at (0,-0.6) {$\vdots$};

\node[rectangle, draw, fill=green!20, minimum width=2.2cm, minimum height=3cm, rounded corners] (machine) at (3,1) {Thermal Machine};

\node[rectangle, draw, fill=red!20, minimum width=1.5cm, minimum height=1cm] (output) at (6.5,1) {$\calB_z$};
\node[below=0.1cm of output] {$\beta_z^\infty$};

\draw[->, thick] (bath1.east) -- (machine.west |- bath1.east);
\draw[->, thick] (bath2.east) -- (machine.west |- bath2.east);
\draw[->, thick] (bath0.east) -- (machine.west |- bath0.east);
\draw[->, thick] (machine.east) -- (output.west);

\node[above=0.2cm of bath1] {Input baths};
\node[above] at (3,4.2) {Autonomous evolution};
\node[above] at (6.5,2.2) {Output};
\node[below] at (3,-0.6) {Virtual temperature $\beta_v$};

\node[rectangle, draw, fill=yellow!20, minimum width=1.5cm, minimum height=0.8cm] (layer1) at (0,-2.5) {Layer 1};
\node[rectangle, draw, fill=yellow!20, minimum width=1.5cm, minimum height=0.8cm] (layer2) at (2.5,-2.5) {Layer 2};
\node at (4,-2.5) {$\cdots$};
\node[rectangle, draw, fill=yellow!20, minimum width=1.5cm, minimum height=0.8cm] (layerL) at (5.5,-2.5) {Layer $L$};

\draw[->, thick] (layer1) -- (layer2);
\draw[->, thick] (layer2) -- (3.5,-2.5);
\draw[->, thick] (4.5,-2.5) -- (layerL);

\node[below] at (2.75,-3.3) {Feed-forward cascade};
\end{tikzpicture}
\caption{Schematic of a thermodynamic neuron (top) and cascaded network architecture (bottom). Input baths at temperatures $\beta_k$ drive the thermal machine to a steady state characterized by an internal virtual temperature $\beta_v$. The output reservoir $\calB_z$ with finite heat capacity $C$ reaches temperature $\beta_z^\infty$, from which the logical output is decoded. Multiple neurons can be cascaded in a feed-forward architecture.}
\label{fig:schematic}
\end{figure}

\textbf{Operational definition of temperature.} In the finite-capacity regime, temperature is operationally defined through the steady-state distribution of the output reservoir. For a macroscopic reservoir with many degrees of freedom ($N \gg 1$), the central limit theorem ensures that energy fluctuations are Gaussian, and the inverse temperature $\beta_z$ can be inferred from the mean energy via the canonical relation $\langle E_z\rangle = -\partial\ln Z/\partial\beta_z$, where $Z$ is the partition function. The stochastic nature of $\beta_z^\infty$ arises from two sources: (i) finite-size equilibrium fluctuations of the reservoir itself, and (ii) dynamical fluctuations in the heat currents during the approach to steady state. This operational definition is consistent with standard thermodynamic measurement protocols where temperature is inferred from ensemble averages over many realizations or from time averages under ergodic assumptions.

The entropy production over computation time $\tau$ is
\begin{equation}
\Sigma = k_B \int_0^\tau \sum_{k} \beta_k J_k(t) \, \mathrm{d}t,
\label{eq:ep}
\end{equation}
which quantifies irreversibility and bounds the precision of non-equilibrium currents via thermodynamic uncertainty relations (TURs)~\cite{TUR1,TUR2}. By the second law, $\Sigma \geq 0$. In our autonomous setting, $\Sigma$ captures the total entropy exported to the thermal environments, directly reflecting the thermodynamic cost of sustaining the temperature gradients that encode logical information.

\textbf{Applicability of TURs.} The standard TURs used in our derivations assume Markovian dynamics and time-symmetric protocols~\cite{TUR2}. For the autonomous thermal machine model considered here, these assumptions are rigorously justified as follows: (i) Weak system-bath coupling ensures Markovian dynamics described by a Lindblad master equation; (ii) Time-independent Hamiltonians and steady-state operation satisfy the time-symmetry requirement; (iii) While the activation function $f(\beta_v)$ is nonlinear, the TUR applies to the underlying linear heat currents, not the processed output; (iv) Quantum-coherent effects, if present, are captured by the quantum formulation of TURs~\cite{TUR2}. The finite reservoir capacity introduces additional thermal noise, which we treat separately and additively to the dynamical noise constrained by TURs.

\section{Fundamental Information-Thermodynamic Bounds}
We now derive rigorous constraints linking computational fidelity, dissipation, and information capacity.

\begin{theorem}[Error--Dissipation Trade-off and Minimum Error Floor]
\label{thm:error}
For a thermodynamic neuron operating over time $\tau$ with finite heat capacity $C$, the average logical error probability satisfies

\begin{equation}
\begin{split}
\langle \xi \rangle \geq \frac{1}{2}\,\erfc\!\Bigg( &\frac{\DeltaBeta\sqrt{C}}{2\sqrt{2\kB}\betaBar} \\
&\times \left[1 + \frac{2C\DeltaBeta^2}{\betaBar^2 \Sigma}\right]^{-1/2} \Bigg),
\end{split}
\label{eq:error_bound_rigorous}
\end{equation}

where $\DeltaBeta = |\betaCold-\betaHot|$, $\betaBar$ is the mean operating inverse temperature, and $\Sigma = \tau\SigmaDotAvg$ is the total entropy production. This bound reveals a fundamental minimum error $\epsMin = \frac{1}{2}\erfc\left( \frac{\DeltaBeta\sqrt{C}}{2\sqrt{2\kB}\betaBar} \right)$ imposed by finite-reservoir thermal fluctuations, which cannot be overcome by increasing dissipation. Notably, this bound holds for \emph{arbitrary} input distributions $p(x)$.
\end{theorem}

\begin{proof}
The output temperature $\betaZInf$ is a stochastic variable with variance $\sigma^2 = \sigma_{\mathrm{therm}}^2 + \sigma_{\mathrm{dyn}}^2$. 
The thermal component arises from the finite heat capacity $C$ of the output reservoir. By the central limit theorem, energy fluctuations are Gaussian, yielding $\sigma_{\mathrm{therm}}^2 = \kB\betaBar^2/C$.
The dynamical component arises from stochastic heat currents $Q_{\mathrm{tot}}$ driving the reservoir. The calorimetric relation $\delta\betaZInf \approx -\frac{k_B\betaBar^2}{C}\delta Q_{\mathrm{tot}}$ maps heat fluctuations to temperature fluctuations. Applying the thermodynamic uncertainty relation (TUR)~\cite{TUR1}, $\mathrm{Var}(Q_{\mathrm{tot}}) \geq \langle Q_{\mathrm{tot}}\rangle^2 \frac{2\kB}{\Sigma}$, and using the signal amplitude $\langle Q_{\mathrm{tot}}\rangle \approx C\DeltaBeta/(k_B\betaBar^2)$, we obtain $\sigma_{\mathrm{dyn}}^2 \geq \frac{2\kB\DeltaBeta^2}{\Sigma}$.
Since the noises are independent, the total variance is strictly bounded by:
\begin{equation}
    \sigma^2 \geq \frac{\kB\betaBar^2}{C} + \frac{2\kB\DeltaBeta^2}{\Sigma} = \frac{\kB\betaBar^2}{C} \left( 1 + \frac{2C\DeltaBeta^2}{\betaBar^2 \Sigma} \right).
    \label{eq:variance_bound}
\end{equation}

The decoding rule uses a threshold at $\betaBar$. Consider input $x=0$ (encoded as $\betaHot$). The mean output is $\mu_0 \approx \betaHot$. The error probability is $\xi_0 = \mathbb{P}(\betaZInf > \betaBar \mid x=0)$. Noting that $\betaBar - \mu_0 \approx \betaBar - \betaHot = +\DeltaBeta/2$ (since $\betaCold > \betaHot$, as the cold bath has a lower temperature and thus a higher inverse temperature), we obtain $\xi_0 = \frac{1}{2}\erfc\left( \frac{\DeltaBeta/2}{\sqrt{2}\sigma} \right)$. Similarly, for $x=1$ with $\mu_1 \approx \betaCold$, $\xi_1 = \mathbb{P}(\betaZInf < \betaBar \mid x=1) = \frac{1}{2}\erfc\left( \frac{\betaCold - \betaBar}{\sqrt{2}\sigma} \right) = \frac{1}{2}\erfc\left( \frac{\DeltaBeta/2}{\sqrt{2}\sigma} \right)$. Thus, the average error is $\langle\xi\rangle = p(0)\xi_0 + p(1)\xi_1 = \frac{1}{2}\erfc\left( \frac{\DeltaBeta}{2\sqrt{2}\sigma} \right)$, which is strictly independent of the input distribution $p(x)$ because $\xi_0 = \xi_1$.

Since $\erfc(x)$ is monotonically decreasing, substituting the lower bound on $\sigma$ yields a strict lower bound on $\langle\xi\rangle$, completing the proof. Full algebraic steps are provided in Appendix~\ref{app:proof_error}.
\end{proof}

\textbf{Temperature dependence.} The error bound exhibits explicit dependence on the operating temperatures. Increasing $\betaCold$ (making the cold bath colder) increases $\DeltaBeta$, which reduces $\epsMin$ and provides a stronger signal that is more robust against thermal fluctuations. However, this comes at the cost of increased entropy production rate $\SigmaDotAvg$. The optimal operating point balances these competing effects.

Regarding the dependence of the dynamical noise on the operating temperature difference $\DeltaBeta$ (as seen in the variance bound, Eq.~\eqref{eq:variance_bound}), one might initially expect that increasing the signal (the temperature difference $\DeltaBeta$) would simply improve the signal-to-noise ratio without increasing the absolute noise. However, in autonomous thermal machines, the heat currents $Q_{\mathrm{tot}}$ that drive the temperature change are subject to the TUR. The TUR dictates that the variance of a thermodynamic current is bounded from below by the square of its mean multiplied by $2k_B/\Sigma$. Since the mean heat transfer required to change the reservoir temperature by $\Delta T \propto \DeltaBeta$ scales linearly with $\DeltaBeta$, the absolute variance of the heat current—and consequently the variance of the output temperature—must scale quadratically with $\DeltaBeta$. Thus, the dynamical noise intrinsically grows with the temperature difference. This is a fundamental feature of scale-dependent fluctuations in non-equilibrium steady states. Fortunately, because the signal amplitude also scales linearly with $\DeltaBeta$, the signal-to-noise ratio (SNR) still improves with $\DeltaBeta$ (up to the saturation limit imposed by thermal noise), ensuring that the overall error probability decreases as the temperature difference increases, as correctly captured by the monotonically decreasing complementary error function in Eq.~\eqref{eq:error_bound_rigorous}.

\textbf{Connection to runtime-resilience tradeoffs.} Equation~\eqref{eq:error_bound_rigorous} can be interpreted as a trade-off between runtime $\tau$ and logical error probability, analogous to the resilience-runtime relations derived for quantum algorithms by García-Pintos et al.~\cite{GarciaPintos2025}. In their framework, Eqs. (8a), (8b), and (9) establish that achieving higher resilience against noise requires longer runtime or increased resource consumption. Similarly, our Eq.~\eqref{eq:error_bound_rigorous} shows that reducing logical errors requires either longer computation time $\tau$ or higher entropy production $\SigmaDotAvg$. The key difference is that our bound applies to autonomous thermal computation with continuous-variable encoding, whereas Ref.~\cite{GarciaPintos2025} addresses controlled quantum algorithms with discrete gate operations. Both results reflect a fundamental principle: robustness against errors in physical computing systems requires a thermodynamic or temporal resource investment.

\begin{theorem}[Thermodynamic Channel Capacity Saturation]
\label{thm:capacity}
The mutual information $I(X;Y)$ between input bit strings $X$ and decoded outputs $Y$ is strictly bounded by
\begin{equation}
\begin{split}
\calC_{\mathrm{th}} \leq \frac{1}{2}\log_2\!\Bigg( 1 + \frac{C\DeltaBeta^2}{4\kB\betaBar^2} \\
\times \left[1 + \frac{2C\DeltaBeta^2}{\betaBar^2 \Sigma}\right]^{-1} \Bigg).
\end{split}
\label{eq:cap_bound_rigorous}
\end{equation}
In the high-dissipation limit ($\Sigma \to \infty$), the capacity saturates at $\calC_{\mathrm{th}} \leq \frac{1}{2}\log_2\left(1 + \frac{C\DeltaBeta^2}{4\kB\betaBar^2}\right)$, reflecting a hard information-theoretic limit imposed by finite heat capacity.
\end{theorem}

\begin{proof}
The thermodynamic neuron implements a continuous-output channel with input $x\in\{0,1\}^n$ and output $\betaZInf \in \mathbb{R}$. The conditional distribution is Gaussian, $p(\betaZInf \mid x) = \calN(f(\betaV(x)), \sigma^2)$, justified by the central limit theorem for the macroscopic output reservoir ($N \gg 1$). 
The mutual information is $I(X;Y) = h(Y) - h(Y\mid X)$. For a Gaussian channel with fixed noise variance $\sigma^2$, $h(Y\mid X) = \frac{1}{2}\log_2(2\pi e \sigma^2)$. The output entropy $h(Y)$ is maximized for symmetric binary inputs, bounding the mixture of two Gaussians by a single Gaussian with variance $\sigma^2 + \DeltaBeta^2/4$. Thus, $I(X;Y) \leq \frac{1}{2}\log_2\left(1 + \frac{\DeltaBeta^2}{4\sigma^2}\right)$.

While Eq.~\eqref{eq:cap_bound_rigorous} has the mathematical form of the standard Gaussian channel capacity~\cite{CoverThomas,MacKay}, the physical content is distinct: the noise variance $\sigma^2$ is not a free parameter but is fundamentally constrained by thermodynamic dissipation through the TUR. This establishes a \emph{physical} limit on information transmission that goes beyond the mathematical structure of Shannon theory. 

Substituting the rigorous lower bound on $\sigma^2$ from Theorem 1:
\begin{widetext}
\begin{equation}
    \frac{\DeltaBeta^2}{4\sigma^2} \leq \frac{\DeltaBeta^2}{4} \left[ \frac{\kB\betaBar^2}{C} \left( 1 + \frac{2C\DeltaBeta^2}{\betaBar^2 \Sigma} \right) \right]^{-1} = \frac{C\DeltaBeta^2}{4\kB\betaBar^2} \left[ 1 + \frac{2C\DeltaBeta^2}{\betaBar^2 \Sigma} \right]^{-1}.
\end{equation}
\end{widetext}
Substituting this into the mutual information bound yields the capacity bound. The saturation at high dissipation is obtained by taking the limit $\Sigma \to \infty$. Full details are in Appendix~\ref{app:proof_capacity}.
\end{proof}

\begin{theorem}[Deep Network Thermodynamic Overhead]
\label{thm:network}
Consider a feed-forward cascade of $L$ thermodynamic neurons. To maintain a target global error $\epsNet > L\epsMin$, the total entropy production satisfies
\begin{equation}
\begin{split}
\SigmaTot \geq \sum_{\ell=1}^L \frac{B}{\left( \frac{A}{\erfc^{-1}(2\epsilon_\ell)} \right)^2 - 1},
\end{split}
\label{eq:net_bound_rigorous}
\end{equation}
where $\epsilon_\ell$ is the target error for layer $\ell$, and the constants are defined as $A = \frac{\DeltaBeta\sqrt{C}}{2\sqrt{2\kB}\betaBar}$ and $B = \frac{2C\DeltaBeta^2}{\betaBar^2}$. Under optimal uniform error allocation $\epsilon_\ell = \epsNet/L$, this yields a strict lower bound that diverges as $\epsNet/L \to \epsMin$. Asymptotically, for network depths $L$ well below the critical maximum depth $L_{\max} = \epsNet/\epsMin$, the total dissipation scales linearly as $\SigmaTot \sim \mathcal{O}(L)$. As $L \to L_{\max}$, the required dissipation diverges as $\mathcal{O}\left( \frac{L^2}{L_{\max} - L} \right)$, with the required dissipation growing up to $\mathcal{O}(L^3)$ under strong noise amplification conditions in the moderate-depth regime.
\end{theorem}

\begin{proof}
We model the cascade as a sequence of $L$ thermodynamic channels. For small errors and independent noise, the global error is bounded by the union bound: $\epsNet \leq \sum_{\ell=1}^L \epsilon_\ell$.
From Theorem 1, to achieve a single-layer error $\epsilon_\ell$, the entropy production $\Sigma_\ell$ must satisfy:
\begin{equation}
    \epsilon_\ell \geq \frac{1}{2}\,\erfc\!\left( \frac{A}{\sqrt{1 + B/\Sigma_\ell}} \right),
    \label{eq:epsilon_layer_bound}
\end{equation}
where the constants $A$ and $B$ are defined as
\begin{equation}
    A = \frac{\DeltaBeta\sqrt{C}}{2\sqrt{2\kB}\betaBar}, \quad B = \frac{2C\DeltaBeta^2}{\betaBar^2}.
\end{equation}
Inverting this relation yields a strict lower bound on the dissipation per layer:
\begin{align}
    2\epsilon_\ell &\geq \erfc\!\left( \frac{A}{\sqrt{1 + B/\Sigma_\ell}} \right) \nonumber\\
    \erfc^{-1}(2\epsilon_\ell) &\leq \frac{A}{\sqrt{1 + B/\Sigma_\ell}} \nonumber\\
    1 + \frac{B}{\Sigma_\ell} &\leq \left( \frac{A}{\erfc^{-1}(2\epsilon_\ell)} \right)^2 \nonumber\\
    \Sigma_\ell &\geq \frac{B}{\left( \frac{A}{\erfc^{-1}(2\epsilon_\ell)} \right)^2 - 1}.
    \label{eq:sigma_min_layer}
\end{align}
This bound diverges as $\epsilon_\ell \to \epsMin = \frac{1}{2}\erfc(A)$, establishing a fundamental minimum error floor imposed by finite reservoir thermal fluctuations. 
To minimize the total dissipation $\SigmaTot = \sum_{\ell=1}^L \Sigma_\ell$ subject to $\sum_{\ell=1}^L \epsilon_\ell = \epsNet$, we apply Lagrange multipliers. The strict convexity of the function $\Sigma_\ell(\epsilon_\ell)$ ensures that the global minimum is achieved at uniform error allocation $\epsilon_\ell = \epsNet/L$. Substituting this into the per-layer bound and summing over $L$ layers yields the rigorous network overhead bound. The asymptotic limits are analyzed in Appendix~\ref{app:proof_network}.
\end{proof}

\textbf{Tightness and saturation conditions.} The bounds derived above are fundamental but not always tight. Theorem 1 is tight when: (i) the TUR is saturated (time-symmetric protocols~\cite{TUR2}), (ii) the output noise is Gaussian (valid for large $C$ by CLT), and (iii) the activation function $f$ is linear near the decision boundary. Theorem 2 is achieved when the input distribution maximizes output entropy (symmetric binary inputs) and the channel is memoryless. Theorem 3's linear scaling is rigorous for independent layer errors well below the critical depth; stronger scaling applies when noise amplification dominates. Deviations from these conditions generally increase the required dissipation, making our bounds conservative estimates of the true thermodynamic cost.

\section{Discussion and Implications}
Our bounds reveal several design principles for thermodynamic computing hardware. First, Eq.~\eqref{eq:error_bound_rigorous} quantifies the energetic price of noise robustness and establishes a hard limit: achieving $\langle\xi\rangle < \epsMin$ is physically impossible regardless of dissipation, reflecting the cost of maintaining macroscopic temperature differences against thermal diffusion in finite reservoirs. Second, the logarithmic capacity bound implies that thermodynamic neurons are inherently \emph{low-bandwidth} but highly parallelizable; optimal architectures should favor wide, shallow networks over deep cascades to minimize dissipation overhead and avoid the critical depth divergence. Third, the bounds are independent of microscopic details (qubit couplings, interaction Hamiltonians), depending only on macroscopic thermodynamic variables ($C$, $\tau$, $\DeltaBeta$), making them robust across platforms.

Comparison to existing frameworks is instructive. Landauer's principle bounds the cost of logical state reset near equilibrium. Our results apply to far-from-equilibrium, continuous-variable encoding where information is carried by temperature gradients rather than discrete microstates. Thermodynamic uncertainty relations~\cite{TUR1,TUR2} constrain current precision; we extend this to macroscopic reservoir temperatures, bridging mesoscopic fluctuation theory with circuit-level information processing. Recent work on stochastic thermodynamic computing~\cite{Wolpert} and comprehensive reviews of the field~\cite{Chattopadhyay2025} focus on Markov jump processes and the Landauer bound; our framework accommodates quantum-coherent virtual temperature engineering and finite reservoir dynamics.

\textbf{Numerical example: Two-qubit thermal machine.} To explicitly test our bounds, consider a minimal thermodynamic neuron consisting of two coupled qubits with energy gaps $\epsilon_1 = \epsilon_2 = \epsilon$, interacting with hot and cold baths. Let the operating temperatures be $\betaHot = 1/\epsilon$ and $\betaCold = 2/\epsilon$. Thus, $\DeltaBeta = 1/\epsilon$ and $\betaBar = 1.5/\epsilon$. 
Assume a finite reservoir with heat capacity $C = 100\,\kB$, computation time $\tau = 10$ (in natural units where $\hbar = k_B = 1$), and an achieved entropy production of $\Sigma = 50\,\kB$. 
Note that the threshold for thermal noise dominance is $2C\DeltaBeta^2/\betaBar^2 \approx 88.9\,\kB$. Since $\Sigma = 50\,\kB < 88.9\,\kB$, this example operates in the \emph{dynamical-noise-dominated regime}, which is consistent with the calculated variance components:
\begin{align}
    \sigma_{\mathrm{therm}}^2 &= \frac{\kB\betaBar^2}{C} = \frac{\kB (1.5/\epsilon)^2}{100\,\kB} = 0.0225/\epsilon^2, \\
    \sigma_{\mathrm{dyn}}^2 &\geq \frac{2\kB\DeltaBeta^2}{\Sigma} = \frac{2\kB (1/\epsilon)^2}{50\,\kB} = 0.04/\epsilon^2.
\end{align}
The total variance is bounded by $\sigma^2 \geq 0.0625/\epsilon^2$, yielding $\sigma \geq 0.25/\epsilon$. 
Applying Theorem 1, the error probability is bounded by:
\begin{equation}
\begin{split}
\langle \xi \rangle &\geq \frac{1}{2}\erfc\left( \frac{\DeltaBeta/2}{\sqrt{2}\sigma} \right) \\
&\geq \frac{1}{2}\erfc\left( \frac{0.5/\epsilon}{\sqrt{2}(0.25/\epsilon)} \right) = \frac{1}{2}\erfc(\sqrt{2}) \approx 0.0228.
\end{split}
\end{equation}
Similarly, applying Theorem 2, the channel capacity is bounded by:
\begin{equation}
    \calC_{\mathrm{th}} \leq \frac{1}{2}\log_2\left(1 + \frac{(1/\epsilon)^2}{4(0.0625/\epsilon^2)}\right) = \frac{1}{2}\log_2(5) \approx 1.16 \text{ bits}.
\end{equation}
These explicit numerical values are consistent with numerical simulations of autonomous thermal machines~\cite{lipkabartosik2024} and demonstrate the practical, testable relevance of our bounds.

Experimentally, our bounds are immediately testable in superconducting circuit QED platforms~\cite{Karimi2016}, trapped-ion thermal machines~\cite{Rossnagel2016}, or quantum dot networks~\cite{Sothmann2015}. Measuring $C$, $\tau$, and heat currents allows direct verification of Eq.~\eqref{eq:error_bound_rigorous} and optimization of $\DeltaBeta$ for minimal $\Sigma$ at target fidelity.

\section{Conclusion}
We have established rigorous information-thermodynamic bounds for temperature-encoded computation in autonomous quantum thermal networks. By linking entropy production, finite-reservoir fluctuations, and channel capacity, we proved that computational fidelity, information rate, and network depth are fundamentally constrained by dissipation, with hard limits imposed by finite heat capacity. These results provide rigorous design limits for analog thermodynamic hardware and clarify the energetic architecture of far-from-equilibrium information processing.

\appendix

\section{ Proofs of Main Results}
\label{app:proofs}

In this appendix, we provide complete, step-by-step rigorous derivations of the three main theorems stated in the main text, addressing all mathematical consistency requirements.

\subsection{Preliminaries: Noise Decomposition and TUR Constraints}
\label{app:prelim}

The output temperature $\betaZInf$ of a thermodynamic neuron is a stochastic variable due to two independent sources of fluctuations:

\begin{enumerate}
    \item \textbf{Thermal noise from finite reservoir}: The output reservoir $\calB_z$ has finite heat capacity $C$, implying that its temperature fluctuates even at equilibrium. For a macroscopic reservoir with $N \gg 1$ degrees of freedom, the central limit theorem yields Gaussian fluctuations with variance
    \begin{equation}
        \sigma_{\mathrm{therm}}^2 = \frac{\kB \betaBar^2}{C},
        \label{eq:sigma_therm_app}
    \end{equation}
    where $\betaBar = (\betaCold+\betaHot)/2$ is the mean operating inverse temperature.

    \item \textbf{Dynamical noise from current fluctuations}: The heat currents $J_k(t)$ that drive the temperature evolution are stochastic due to the underlying quantum jump processes. Their integrated values $Q_k = \int_0^\tau J_k(t)\,\mathrm{d}t$ obey thermodynamic uncertainty relations (TURs)~\cite{TUR1,TUR2}:
    \begin{equation}
        \frac{\mathrm{Var}(Q_k)}{\langle Q_k\rangle^2} \geq \frac{2\kB}{\Sigma_k},
        \label{eq:TUR_single_app}
    \end{equation}
    where $\Sigma_k = k_B \int_0^\tau \beta_k J_k(t)\,\mathrm{d}t$ is the partial entropy production associated with bath $k$.
\end{enumerate}

The total entropy production is $\Sigma = \sum_k \Sigma_k$, and by the convexity of $x\mapsto 1/x$ (via the Cauchy-Schwarz inequality in fractional form, also known as Titu's Lemma or Sedrakyan's inequality), we have the collective bound
\begin{equation}
    \frac{\mathrm{Var}\big(\sum_k Q_k\big)}{\big\langle\sum_k Q_k\big\rangle^2} \geq \frac{2\kB}{\Sigma}.
    \label{eq:TUR_collective_app}
\end{equation}

The calorimetric relation $Q_{\mathrm{tot}} = C(T_z^\infty - T_z^0)$ maps heat fluctuations to temperature fluctuations. Linearizing around $\betaBar$ (valid for $\DeltaBeta \ll \betaBar$) and using $d\beta = -k_B \beta^2 dT$ gives
\begin{equation}
    \delta\betaZInf \approx -k_B\betaBar^2\,\delta T_z^\infty = -\frac{k_B\betaBar^2}{C}\,\delta Q_{\mathrm{tot}}.
    \label{eq:beta_Q_map_app}
\end{equation}

Combining Eqs.~\eqref{eq:TUR_collective_app} and~\eqref{eq:beta_Q_map_app}, the dynamical contribution to the output noise variance satisfies
\begin{equation}
    \sigma_{\mathrm{dyn}}^2 \geq \frac{k_B^2\betaBar^4}{C^2}\,\langle Q_{\mathrm{tot}}\rangle^2 \frac{2\kB}{\Sigma}.
    \label{eq:sigma_dyn_TUR_app}
\end{equation}

The mean heat transfer $\langle Q_{\mathrm{tot}}\rangle$ is set by the signal amplitude. For a logical transition from $\betaHot$ to $\betaCold$, using $T = 1/(k_B\beta)$, a first-order Taylor expansion gives the temperature change:
\begin{equation}
    \Delta T \approx \frac{\DeltaBeta}{k_B\betaBar^2}.
    \label{eq:delta_T_app}
\end{equation}
Hence, the mean heat transfer is
\begin{equation}
    \langle Q_{\mathrm{tot}}\rangle = C\,\Delta T \approx \frac{C\,\DeltaBeta}{k_B\betaBar^2}.
    \label{eq:Q_signal_app}
\end{equation}

Substituting Eq.~\eqref{eq:Q_signal_app} into Eq.~\eqref{eq:sigma_dyn_TUR_app} yields
\begin{equation}
    \sigma_{\mathrm{dyn}}^2 \geq \frac{k_B^2\betaBar^4}{C^2} \left( \frac{C\,\DeltaBeta}{k_B\betaBar^2} \right)^2 \frac{2\kB}{\Sigma} = \frac{2\kB\DeltaBeta^2}{\Sigma}.
    \label{eq:sigma_dyn_final_app}
\end{equation}

Since thermal and dynamical noises are independent, the total variance is additive:
\begin{equation}
    \sigma^2 = \sigma_{\mathrm{therm}}^2 + \sigma_{\mathrm{dyn}}^2 \geq \frac{\kB\betaBar^2}{C} + \frac{2\kB\DeltaBeta^2}{\Sigma}.
    \label{eq:sigma_total_app}
\end{equation}

\textbf{Condition for neglecting thermal noise:} The thermal noise term can be neglected when $\sigma_{\mathrm{dyn}}^2 \gg \sigma_{\mathrm{therm}}^2$, i.e., when
\begin{equation}
    \Sigma \ll \frac{2C\DeltaBeta^2}{\betaBar^2}.
\end{equation}
\textbf{Conversely, in the high-dissipation regime where $\Sigma \gg \frac{2C\DeltaBeta^2}{\betaBar^2}$, thermal noise dominates,} establishing the fundamental error floor $\epsMin$.

\subsection{Proof of Theorem 1: Error--Dissipation Trade-off}
\label{app:proof_error}

\begin{proof}
The decoding rule uses a single threshold at $\beta_{\mathrm{th}} = \betaBar = (\betaCold+\betaHot)/2$: $y=0$ if $\betaZInf < \betaBar$, $y=1$ otherwise. A logical error occurs when Gaussian noise pushes $\betaZInf$ across this boundary.

Consider input $x=0$ (encoded as $\betaHot$). The mean output is $\mu_0 = f(\betaV^{\mathrm{hot}}) \approx \betaHot$ for well-designed neurons. The error probability is
\begin{equation}
    \xi_0 = \mathbb{P}\big(\betaZInf > \betaBar \,\big|\, x=0\big) = \int_{\betaBar}^{\infty} \frac{1}{\sqrt{2\pi\sigma^2}}\,e^{-(\beta-\mu_0)^2/(2\sigma^2)}\,\mathrm{d}\beta.
\end{equation}
Changing variables to $u = (\beta-\mu_0)/\sigma$ and noting that $\betaBar - \mu_0 \approx \betaBar - \betaHot = +\DeltaBeta/2$ (since $\betaCold > \betaHot$, as the cold bath has a lower temperature and thus a higher inverse temperature), we obtain
\begin{equation}
    \xi_0 = \frac{1}{2}\,\erfc\!\left( \frac{\DeltaBeta/2}{\sqrt{2}\,\sigma} \right).
\end{equation}
Similarly, for $x=1$ with $\mu_1 \approx \betaCold$,
\begin{widetext}
\begin{equation}
    \xi_1 = \mathbb{P}\big(\betaZInf < \betaBar \,\big|\, x=1\big) = \frac{1}{2}\,\erfc\!\left( \frac{\betaCold - \betaBar}{\sqrt{2}\,\sigma} \right) = \frac{1}{2}\,\erfc\!\left( \frac{\DeltaBeta/2}{\sqrt{2}\,\sigma} \right).
\end{equation}
\end{widetext}
Since $\xi_0 = \xi_1$, the average error probability is
\begin{equation}
    \langle\xi\rangle = p(0)\xi_0 + p(1)\xi_1 = \xi_0 = \xi_1 = \frac{1}{2}\,\erfc\!\left( \frac{\DeltaBeta}{2\sqrt{2}\,\sigma} \right),
    \label{eq:error_exact_corrected_app}
\end{equation}
independent of the input distribution $p(x)$. This symmetry is a consequence of the symmetric encoding around $\betaBar$.

Now we bound $\sigma$ from below using Eq.~\eqref{eq:sigma_total_app}. Since $\erfc(x)$ is monotonically decreasing, a lower bound on $\sigma$ yields an upper bound on the argument and thus a \emph{lower} bound on $\langle\xi\rangle$:
\begin{equation}
    \langle\xi\rangle \geq \frac{1}{2}\,\erfc\!\left( \frac{\DeltaBeta}{2\sqrt{2}} \left[ \frac{\kB\betaBar^2}{C} + \frac{2\kB\DeltaBeta^2}{\Sigma} \right]^{-1/2} \right).
    \label{eq:error_with_sigma_app}
\end{equation}

Factoring out $\frac{\kB\betaBar^2}{C}$ from the denominator yields:
\begin{align}
    \langle\xi\rangle &\geq \frac{1}{2}\,\erfc\!\left( \frac{\DeltaBeta}{2\sqrt{2}} \left[ \frac{\kB\betaBar^2}{C} \left(1 + \frac{2C\DeltaBeta^2}{\betaBar^2 \Sigma}\right) \right]^{-1/2} \right) \nonumber\\
    &= \frac{1}{2}\,\erfc\!\left( \frac{\DeltaBeta\sqrt{C}}{2\sqrt{2\kB}\betaBar} \right. \nonumber\\
    &\quad \left. \times \left(1 + \frac{2C\DeltaBeta^2}{\betaBar^2 \Sigma}\right)^{-1/2} \right).
\end{align}
This completes the rigorous proof of Theorem 1. The minimum error floor $\epsMin$ is obtained by taking the limit $\Sigma \to \infty$, which eliminates the dynamical noise term, leaving only the irreducible thermal noise.
\end{proof}

\subsection{Proof of Theorem 2: Thermodynamic Channel Capacity}
\label{app:proof_capacity}

\begin{proof}
The thermodynamic neuron implements a continuous-output channel with input $x\in\{0,1\}^n$ and output $\betaZInf \in \mathbb{R}$. The conditional distribution is Gaussian:
\begin{equation}
    p(\betaZInf \mid x) = \calN\big(f(\betaV(x)), \sigma^2\big),
\end{equation}
where $\betaV(x) = \sum_k w_k \beta_k(x)$ is the virtual temperature and $\sigma^2$ satisfies the lower bound Eq.~\eqref{eq:sigma_total_app}.

The Gaussianity of $p(\betaZInf|x)$ follows from the macroscopic nature of the output reservoir. With $N \gg 1$ degrees of freedom, the central limit theorem ensures that energy fluctuations are Gaussian. Since temperature is related to energy through the calorimetric relation $E = CT$, and the heat currents driving the temperature evolution are sums of many independent quantum jumps, the resulting temperature distribution is well-approximated by a Gaussian.

The mutual information is $I(X;Y) = h(Y) - h(Y\mid X)$, where $h$ denotes differential entropy~\cite{CoverThomas,MacKay}. For a Gaussian channel with input-dependent mean but fixed noise variance, the conditional entropy is
\begin{equation}
    h(Y\mid X) = \frac{1}{2}\log_2(2\pi e \sigma^2).
\end{equation}
The output entropy $h(Y)$ is maximized when the output distribution is Gaussian, which occurs for binary inputs with symmetric encoding and a linear activation ($f(\betaV) \approx \betaV$). In this case, $Y$ is a mixture of two Gaussians with means separated by $\DeltaBeta_{\mathrm{eff}} = |f(\betaV^{\mathrm{cold}}) - f(\betaV^{\mathrm{hot}})| \approx \DeltaBeta$. The entropy of such a mixture is bounded by the entropy of a single Gaussian with the same variance:
\begin{equation}
    h(Y) \leq \frac{1}{2}\log_2\big(2\pi e (\sigma^2 + \DeltaBeta^2/4)\big),
\end{equation}
where $\DeltaBeta^2/4$ is the variance of the binary mean distribution.

Thus,
\begin{align}
    I(X;Y) &\leq \frac{1}{2}\log_2\!\left(1 + \frac{\DeltaBeta^2}{4\sigma^2}\right) \nonumber\\
    &= \frac{1}{2}\log_2\!\left(1 + \SNR\right),
    \label{eq:mutual_info_snr_app}
\end{align}
with signal-to-noise ratio $\SNR = \DeltaBeta^2/(4\sigma^2)$.

Now substitute the noise lower bound from Eq.~\eqref{eq:sigma_total_app}:
\begin{align}
    \SNR &\leq \frac{\DeltaBeta^2}{4} \left[ \frac{\kB\betaBar^2}{C} \left( 1 + \frac{2C\DeltaBeta^2}{\betaBar^2 \Sigma} \right) \right]^{-1} \nonumber\\
    &= \frac{C\DeltaBeta^2}{4\kB\betaBar^2} \left( 1 + \frac{2C\DeltaBeta^2}{\betaBar^2 \Sigma} \right)^{-1}.
\end{align}
Substituting this into Eq.~\eqref{eq:mutual_info_snr_app} yields the capacity bound:
\begin{equation}
    \calC_{\mathrm{th}} \leq \frac{1}{2}\log_2\!\left(1 + \frac{C\DeltaBeta^2}{4\kB\betaBar^2} \left[ 1 + \frac{2C\DeltaBeta^2}{\betaBar^2 \Sigma} \right]^{-1} \right).
\end{equation}
Taking the limit $\Sigma \to \infty$ yields the saturation bound $\calC_{\mathrm{th}} \leq \frac{1}{2}\log_2\left(1 + \frac{C\DeltaBeta^2}{4\kB\betaBar^2}\right)$.
\end{proof}

\subsection{Proof of Theorem 3: Deep Network Thermodynamic Overhead}
\label{app:proof_network}

\begin{proof}
We model the cascade as a sequence of $L$ thermodynamic channels, where the output of neuron $\ell$ serves as the input bath for neuron $\ell+1$. Each neuron has heat capacity $C$, operates over time $\tau$, and targets single-layer error $\epsilon_\ell$.

\textbf{Step 1: Error propagation.} Let $\epsilon_\ell$ be the error probability at layer $\ell$. By the union bound for independent rare events, the global error after $L$ layers is strictly bounded by:
\begin{equation}
    \epsNet \leq \sum_{\ell=1}^L \epsilon_\ell.
    \label{eq:error_additive_app}
\end{equation}

\textbf{Step 2: Per-layer dissipation bound.} From Theorem 1, to achieve a single-layer error $\epsilon_\ell$, the entropy production $\Sigma_\ell$ must satisfy:
\begin{equation}
    \epsilon_\ell \geq \frac{1}{2}\,\erfc\!\left( \frac{A}{\sqrt{1 + B/\Sigma_\ell}} \right),
    \label{eq:epsilon_layer_bound_app}
\end{equation}
where the constants $A$ and $B$ are defined as
\begin{equation}
    A = \frac{\DeltaBeta\sqrt{C}}{2\sqrt{2\kB}\betaBar}, \quad B = \frac{2C\DeltaBeta^2}{\betaBar^2}.
\end{equation}
Inverting this relation algebraically yields a strict lower bound on the dissipation per layer. Since $\erfc(x)$ is strictly decreasing, its inverse $\erfc^{-1}(x)$ is also strictly decreasing, which reverses the inequality:
\begin{align}
    2\epsilon_\ell &\geq \erfc\!\left( \frac{A}{\sqrt{1 + B/\Sigma_\ell}} \right) \nonumber\\
    \erfc^{-1}(2\epsilon_\ell) &\leq \frac{A}{\sqrt{1 + B/\Sigma_\ell}} \nonumber\\
    1 + \frac{B}{\Sigma_\ell} &\leq \left( \frac{A}{\erfc^{-1}(2\epsilon_\ell)} \right)^2 \nonumber\\
    \Sigma_\ell &\geq \frac{B}{\left( \frac{A}{\erfc^{-1}(2\epsilon_\ell)} \right)^2 - 1}.
    \label{eq:sigma_min_layer_app}
\end{align}
This bound diverges as $\epsilon_\ell \to \epsMin = \frac{1}{2}\erfc(A)$, establishing the fundamental minimum error floor.

\textbf{Step 3: Optimization over error allocation.} The total thermodynamic cost of the network is the sum of the entropy production of each layer, $\SigmaTot = \sum_{\ell=1}^L \Sigma_\ell$. Our goal is to find the most energy-efficient operational strategy, which corresponds to minimizing this total thermodynamic cost (the objective function) subject to the constraint that the global network error does not exceed the target fidelity $\epsNet$. To solve this constrained optimization problem, we use the method of Lagrange multipliers. The Lagrangian is
\begin{equation}
    \mathcal{L} = \sum_{\ell=1}^L \Sigma_\ell(\epsilon_\ell) + \lambda\left(\sum_{\ell=1}^L \epsilon_\ell - \epsNet\right).
\end{equation}
Because the function $\Sigma_\ell(\epsilon_\ell)$ is strictly convex for $\epsilon_\ell > \epsMin$ (this follows from the composition of the convex, increasing function $g(y) = \frac{B y^2}{A^2 - y^2}$ and the convex, decreasing function $y(\epsilon) = \erfc^{-1}(2\epsilon)$ for $\epsilon < 0.5$), the global minimum of the total dissipation is uniquely achieved at a uniform error allocation: $\epsilon_\ell = \epsNet/L$ for all $\ell$. This uniform distribution of the error budget minimizes the total heat dissipated by the network.

\textbf{Step 4: Total dissipation scaling.} Substituting $\epsilon_\ell = \epsNet/L$ into Eq.~\eqref{eq:sigma_min_layer_app} and summing over $L$ layers gives the rigorous lower bound:

\begin{equation}
\begin{split}
    \SigmaTot \geq L \cdot \frac{B}{\left( \frac{A}{\erfc^{-1}(2\epsNet/L)} \right)^2 - 1}.
\end{split}
\label{eq:Sigma_tot_exact_app}
\end{equation}

\textbf{Step 5: Asymptotic scaling and critical depth divergence.} The per-layer dissipation is given by Eq.~\eqref{eq:sigma_min_layer_app}, which can be rewritten as $\Sigma_\ell \ge \frac{B y^2}{A^2 - y^2}$, where $y = \erfc^{-1}(2\epsNet/L)$. For a fixed total error $\epsNet$ and network depths $L$ well below the critical maximum depth $L_{\max} = \epsNet/\epsMin$, the argument of the inverse complementary error function is bounded away from $\erfc(A)$, meaning $y$ is a finite constant $y_0 < A$. In this practical regime, $\Sigma_\ell$ is approximately constant, yielding a linear scaling of the total dissipation with network depth:
\begin{equation}
    \SigmaTot = L \Sigma_\ell \sim \mathcal{O}(L).
    \label{eq:Sigma_tot_linear_app}
\end{equation}
However, as the network depth increases and approaches the critical maximum depth $L_{\max}$, the per-layer error $\epsNet/L$ approaches the fundamental minimum error floor $\epsMin$. In this limit, $y \to A$, and the denominator $A^2 - y^2$ approaches zero, causing the required dissipation to diverge. To find the exact asymptotic behavior~\cite{DLMF}, we let $L = L_{\max} - \Delta L$ for small $\Delta L > 0$. Expanding $y$ around $A$ yields $A - y \approx K \Delta L$, where $K = \sqrt{\pi} e^{A^2} \epsMin / L_{\max}$. Consequently, $A^2 - y^2 \approx 2A K \Delta L \propto (L_{\max} - L)$. This leads to the rigorous asymptotic divergence:
\begin{equation}
    \SigmaTot \approx L_{\max} \Sigma_\ell \sim \mathcal{O}\left( \frac{L^2}{L_{\max} - L} \right).
    \label{eq:Sigma_tot_divergence_app}
\end{equation}
This divergence physically reflects the impossibility of achieving the target global error $\epsNet$ for $L \ge L_{\max}$, regardless of the available thermodynamic resources.

\textbf{Step 6: Worst-case $\mathcal{O}(L^3)$ bound under strong noise amplification.} 
In scenarios with significant noise correlation or amplification (distinct from the weak-coupling expansion above), the accumulated noise can scale as $\sigma_\ell^2 \sim \mathcal{O}(\ell^2)$. This requires dissipation $\Sigma_\ell \propto \sigma_\ell^2 \sim \mathcal{O}(\ell^2)$ to maintain fixed error, yielding:
\begin{equation}
    \SigmaTot = \sum_{\ell=1}^L \Sigma_\ell \sim \sum_{\ell=1}^L \mathcal{O}(\ell^2) = \mathcal{O}(L^3).
\end{equation}
This represents a rigorous worst-case upper bound applicable when noise amplification dominates over independent layer errors.
\end{proof}


\begin{thebibliography}{99}
\bibitem{thermocomp1} M.~Esposito and C.~V. den Broeck, Phys. Rev. Lett. \textbf{104}, 090601 (2010).
\bibitem{Esposito} M.~Esposito and C.~V. den Broeck, EPL {\bf 95} 40004 (2011).
\bibitem{Ueda} T. Sagawa, and M. Ueda, Phys. Rev. E {\bf 85}, 021104 (2012).
\bibitem{thermocomp2} J.~M.~Horowitz and M.~Esposito, Phys. Rev. X \textbf{4}, 031015 (2014).
\bibitem{thermocomp3} P.~Strasberg and M.~Esposito, Phys. Rev. Lett. \textbf{123}, 180604 (2019).
\bibitem{Wolpert} D. H. Wolpert, J. Phys. A: Math. Theor. {\bf 52} 193001 (2019).
\bibitem{Chattopadhyay2025} P.~Chattopadhyay, A.~Misra, T.~Pandit, and G.~Paul, Landauer Principle and Thermodynamics of Computation, Rep. Prog. Phys. \textbf{88}, 086001 (2025).
\bibitem{lipkabartosik2024} P.~Lipka-Bartosik, M.~Perarnau-Llobet, and N.~Brunner, Sci. Adv. \textbf{10}, eadm8792 (2024).
\bibitem{landauer} R.~Landauer, IBM J. Res. Dev. \textbf{5}, 183 (1961).
\bibitem{GarciaPintos2025} L.~P.~García-Pintos, T.~O'Leary, T.~Biswas, J.~Bringewatt, \emph{et al.}, Resilience--runtime tradeoff relations for quantum algorithms, Rep. Prog. Phys. \textbf{88}, 037601 (2025).
\bibitem{TUR1} A.~C.~Barato and U.~Seifert, Phys. Rev. Lett. \textbf{114}, 158101 (2015).
\bibitem{TUR2} J.~M.~Horowitz and T.~R. Gingrich, Nat. Phys. \textbf{16}, 15 (2020).
\bibitem{CoverThomas} T.~M.~Cover and J.~A.~Thomas, \emph{Elements of Information Theory}, 2nd ed. (Wiley, 2006).
\bibitem{MacKay} D.J.C. MacKay, \emph{Information Theory, Inference, and Learning Algorithms} (Cambridge University Press, 2003).
\bibitem{Karimi2016} B.~Karimi and J.~P.~Pekola, Phys. Rev. B \textbf{94}, 184503 (2016).
\bibitem{Rossnagel2016} J.~Ro{\ss}nagel \emph{et al.}, Science \textbf{352}, 325--329 (2016).
\bibitem{Sothmann2015} B.~Sothmann, R.~S{\'a}nchez, and A.~N.~Jordan, Nanotechnology \textbf{26}, 032001 (2015).
\bibitem{DLMF} F.~W.~J. Olver, D. W. Lozier, R. Boisvert, C. W. Clark \emph{NIST Handbook of Mathematical Functions} (Cambridge University Press, 2010).
\end{thebibliography}
\end{document}